\documentclass[prd,reprint]{revtex4-2}

\usepackage{mathtools}
\usepackage{amsmath,amssymb,amsthm}
\usepackage{bbm}
\usepackage{enumitem}
\theoremstyle{plain}
\newtheorem{theorem}{Theorem}

\theoremstyle{definition}

\usepackage{xcolor}
\definecolor{darkgreen}{rgb}{0,.3,0}
\definecolor{darkblue}{rgb}{0,0,.5}
\definecolor{darkred}{rgb}{.4,0,0}

\usepackage{comment}
\usepackage{wrapfig}

\usepackage[hyphens]{xurl}
\usepackage{hyperref}

\hypersetup{
	pdfauthor={Paul-Hermann Balduf, Simone Hu},
	pdfcreator={LaTeX with hyperref},
	colorlinks=true,
	linkcolor=darkred,
	citecolor=darkgreen,
	filecolor=black,
	urlcolor=darkblue,
	bookmarks=true,
	plainpages=false,
	pdfpagelabels=true,
	breaklinks=true
}

\usepackage[nameinlink]{cleveref}
\crefname{figure}{figure}{figures}
\usepackage{slashed}

\usepackage[font=small,labelfont=bf]{caption}
\usepackage[font=footnotesize]{subcaption}
\usepackage{placeins} 
\usepackage{csquotes} 

\usepackage{graphicx}
\usepackage{tikz}
\usetikzlibrary{arrows.meta}
\usetikzlibrary{calc}
\tikzstyle{edge} = [black,line width=.4mm]
\tikzstyle{momentumarrow}=[blue,line width=.2mm,> = Stealth]
\tikzstyle{vertex}=[circle,minimum size=2.2mm, draw=black, fill=black, inner sep=0mm]

\usepackage[framemethod=tikz,
roundcorner=2pt,
linewidth=.5pt,
skipabove=10pt plus 20pt minus 2pt,
skipbelow=12pt plus 20pt minus 2pt]{mdframed}
\mdfsetup{linewidth=1.0pt, roundcorner=5pt}

\newcommand{\renop} { \mathcal {R}  }
\newcommand{\ren} {{\tiny \mathcal R}}
\newcommand{\loopnumber}{\ell}
\newcommand{\Graph}{G}
\renewcommand{\d}{\textnormal{d}}
\newcommand*{\mb}{\mathbbm}
\newcommand*{\mc}{\mathcal}
\newcommand\scalemath[2]{\scalebox{#1}{\mbox{\ensuremath{\displaystyle #2}}}}
\newcommand{\period}{\mathcal{P}}
\newcommand{\abs}[1]{\left | #1 \right |}

\begin{document}
	
	\title{Ladders and rainbows in Minimal Subtraction}
	
	\date{\today}
	
	\author{Paul-Hermann Balduf}
	
	\email{paul-hermann.balduf@maths.ox.ac.uk}
	\homepage{https://paulbalduf.com}
	\affiliation{Mathematical Institute, University of Oxford, OX2 6GG, UK.}
	
	\begin{abstract}
		In dimensional regularization with $D=D_0-2\epsilon$, the minimal subtraction (MS) scheme is characterized by counterterms that only consist of singular terms in $\epsilon$. We develop a general method to compute the infinite sums of massless ladder or rainbow Feynman integrals in MS at $D_0$. Our method is based on relating the MS-solution to a kinematic solution at a coupling-dependent renormalization point. If the $\epsilon$-dependent Mellin transform of the kernel diagram of the insertions can be computed in closed form, we typically obtain a closed expression for the all-order  solution in MS. As examples, we consider Yukawa theory and $\phi^4$ theory in $D_0=4$, and $\phi^3$ theory in $D_0=6$.
	\end{abstract}

	\maketitle

\section{Introduction}

Many quantum field theories contain infinite families of Feynman diagrams which arise from repeatedly inserting subdiagrams into the same kernel diagram. Two particular such cases are \emph{ladders} and  \emph{rainbows}. Almost 30 years ago, the exact sum of ladders and rainbows for $\phi^3$ and Yukawa theory has been computed   in kinematic (MOM) renormalization conditions  \cite{delbourgo_dimensional_1996,delbourgo_dimensional_1997}. 

From the perspective of Hopf algebra theory of renormalization \cite{kreimer_overlapping_1999,connes_renormalization_2000,connes_renormalization_2001}, the renormalized amplitudes of such sums of diagrams are the solutions to \emph{linear} single-scale Dyson-Schwinger equations. By now, there is a systematic procedure to construct  the solution in the MOM scheme from the Mellin transform of the kernel diagram, developed by Broadhurst, Kreimer, Yeats, and collaborators  \cite{broadhurst_renormalization_1999,broadhurst_exact_2001,kreimer_etude_2006,kreimer_etude_2008,kreimer_recursion_2008,yeats_growth_2008}. Conversely, solutions in the minimal subtraction (MS) scheme have so far only  been computed numerically \cite{broadhurst_renormalization_1999,balduf_dyson_2023}, where in the latter publication, some closed formulas have been discovered empirically by matching their series expansion. 

In the present work, we present a general method to compute the closed-form solution of linear single-scale single-kernel Dyson-Schwinger equations in the MS scheme.  We confirm   the formulas found in \cite{balduf_dyson_2023}, and compute the exact solution for some further examples.

\subsection{Linear Dyson-Schwinger equations}

Ladders and rainbows are infinite families of Feynman diagrams which are characterized by recursively inserting an already existing nested diagram into some fixed \emph{kernel} diagram at every new order in perturbation theory. 
In the present article, we restrict ourselves to the case where the kernel diagram is ultraviolet divergent, free of ultraviolet subdivergences, and free of infrared divergences. This ensures that the class of diagrams obtained this way is closed under perturbative renormalization, that is, they form a sub Hopf algebra in the Hopf algebra of renormalization \cite{kreimer_anatomy_2006,foissy_faa_2008,foissy_classification_2010}. The decisive feature of ladders/rainbows is that only one copy of the existing nested diagram is inserted at each iteration, and it is always inserted into the same position in the kernel diagram. This  implies that the procedure can be described by a \emph{linear} Dyson-Schwinger equation (DSE), schematically of the form 
\begin{align}\label{DSE_combinatorial}
G_\ren &= 1+\alpha \left( 1-\renop \right) B_+ \left[ G_\ren \right] .
\end{align}
Here, $G_\ren$ is a renormalized 1PI Green's function, an infinite formal series of Feynman diagrams. $G_\ren$ is a scalar. If the theory in question has non-trivial tensor structures, $G_\ren$ is understood to be a projection onto a suitable basis tensor, also called form factor.  $\alpha$ is the coupling, the operator $B_+$ denotes insertion into the kernel diagram, and $\renop$ is a renormalization operator.  
We shall clarify the precise meaning of \cref{DSE_combinatorial} in the following, but first, to have a concrete example at hand, we consider the case of rainbows for the 1PI propagator of $\phi^3$ theory.

\begin{figure}[htbp]
	\centering
	\begin{tikzpicture}
		\node at (-1.5,.1){$B_+[G] =$};
		\node[vertex] (v1) at (0,0){};
		\node[vertex](v2) at (3,0){};
		\node(vx) at (1.5,0){$G$};
		\draw[edge,  bend angle=50, bend left] (v1) to (v2);
		\draw[edge] (v1) -- (vx) --(v2);
		\draw[edge] (v1)-- +(-.5,0);
		\draw[edge ] (v2)-- +(.5,0);
	\end{tikzpicture}
	
	\caption{1-loop kernel for the propagator of $\phi^3$ theory. The operator $B_+[G]$ means to compute the Feynman integral of this diagram, where a subdiagram $G$ has been inserted into the lower  edge. Notice that a propagator-type (amputated) subdiagram cancels one of its adjacent edges since it is proportional to an inverse propagator.}
	\label{fig:B_phi3}
\end{figure}
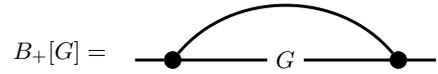

In this case, the kernel diagram is the 1-loop multiedge  shown in \cref{fig:B_phi3}, and the operator $B_+[G]$ denotes insertion of  $G$ (which may be a single diagram or a sum of diagrams, in which case $B_+$ acts linearly) into the lower one of the two internal edges.  In particular, $B_+[\mb 1]$ is the kernel itself without any insertions, 
\begin{align*}
B_+[\mb 1](p) &= \int \frac{\d^D k}{(2\pi)^D} \frac{1}{k^2} \frac{1}{(k-p)^2},\\
 B_+[G](p) &= \int \frac{\d^D k}{(2\pi)^D} \frac{1}{k^2} G(k) \frac{1}{(k-p)^2}.
\end{align*}
Here, we have left the spacetime dimension $D=D_0-2\epsilon$ arbitrary in order to use dimensional regularization. 
For the renormalization operator $\renop$, we may then choose   \emph{minimal subtraction} (MS) renormalization conditions, which means that  $\renop$ projects onto the pole terms in $\epsilon$, such that $(1-\renop)$ subtracts pole terms. Alternatively, we can choose \emph{kinematic renormalization} (MOM), still with $D=D_0-2\epsilon$, then $\renop$  projects onto a fixed value of external kinematic parameters, such that $(1-\renop)$ vanishes at that value. 
Repeatedly inserting the already existing sum $G_\ren$ into the kernel diagram, as described by the DSE in \cref{DSE_combinatorial},  gives rise to the sum of rainbow diagrams shown in \cref{fig:phi3_rainbows}.

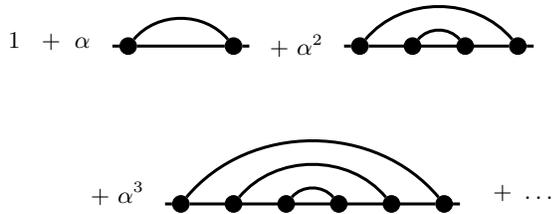
\begin{figure}[htbp]
	\centering
	\begin{tikzpicture}[scale=.7]
		\node at (-1.5,.1){$1 ~~+~ \alpha$};
		\node[vertex] (v1) at (0,0){};
		\node[vertex](v2) at (2,0){};
		
		\draw[edge,  bend angle=50, bend left] (v1) to (v2);
		\draw[edge] (v1) --  (v2);
		\draw[edge] (v1)-- +(-.3,0);
		\draw[edge ] (v2)-- +(.3,0);
		
		\node at (3.2,0){$+~\alpha^2$};
		
		\node[vertex] (v1) at (4.4,0){};
		\node[vertex](v2) at (7.4,0){};
		\node[vertex](v3) at (5.4,0){};
		\node[vertex](v4) at (6.4,0){};
		
		\draw[edge,  bend angle=50, bend left] (v1) to (v2);
		\draw[edge,  bend angle=50, bend left] (v3) to (v4);
		\draw[edge] (v1)  --(v2);
		\draw[edge] (v1)-- +(-.3,0);
		\draw[edge ] (v2)-- +(.3,0);
		
		\node at (-.2,-2.8){$+~\alpha^3$};
		
		\node[vertex] (v1) at (1,-3){};
		\node[vertex](v2) at (6,-3){};
		\node[vertex](v3) at (2,-3){};
		\node[vertex](v4) at (3,-3){};
		\node[vertex](v5) at (4,-3){};
		\node[vertex](v6) at (5,-3){};
		
		\draw[edge,  bend angle=50, bend left] (v1) to (v2);
		\draw[edge,  bend angle=50, bend left] (v3) to (v6);
		\draw[edge,  bend angle=50, bend left] (v4) to (v5);
		\draw[edge] (v1)  --(v2);
		\draw[edge] (v1)-- +(-.3,0);
		\draw[edge ] (v2)-- +(.3,0);
		
		\node at (7.5,-2.8){$+~\ldots$};
	\end{tikzpicture}
	
	\caption{Sum of rainbows in $\phi^3$ theory. }
	\label{fig:phi3_rainbows}
\end{figure}

Obviously, the sum of rainbows only represents a small subset of the diagrams in full $\phi^3$ theory. Conceptually, a linear DSE such as \cref{DSE_combinatorial} is a simplification in multiple ways; a full quantum field theory typically has multiple coupled DSEs (in particular, in a non-scalar theory, a single Green's function can have multiple basis tensor structures), each of them has multiple kernel diagrams, and insertions can happen in more than one place in each kernel. We comment on these generalizations after describing how to solve the single-scale single-kernel equation in MOM.

\subsection{Solution in kinematic renormalization}\label{sec:MOM}

All-order solutions to ladder and rainbow DSEs in MOM were first derived by Delbourgo and collaborators in position space and for arbitary dimension $D$ in \cite{delbourgo_dimensional_1997,delbourgo_dimensional_1996}. 
However, for our solution in minimal subtraction, a different method,  developed by Broadhurst, Kreimer, Yeats and collaborators \cite{broadhurst_renormalization_1999,broadhurst_combinatoric_2000,broadhurst_exact_2001,kreimer_etude_2006,kreimer_etude_2008,kreimer_recursion_2008,yeats_growth_2008}, is more suitable.   This methods has the additional benefits that it has straightforward relations to  conventional momentum-space Feynman integral calculations,   it allows for various generalizations beyond the linear case, and it is fundamentally related to the Hopf algebra theory of renormalization.  We briefly review this method in the following as far as we need it, comprehensive accounts are, for example, the books/theses  \cite{balduf_dyson_2024,yeats_combinatorial_2017,panzer_hopf_2012,olson-harris_applications_2024}.

We consider massless \emph{single-scale} Dyson-Schwinger equations, this means that there is only one kinematic variable. For a propagator-type Green's function, this is the external momentum $p$, for other Green's functions, all but one scale need to be fixed. Under this condition, the solution of every Feynman integral is itself proportional to the kinematic variable. A Feynman diagram  in spacetime dimension $D$ with $\loopnumber$ loops and propagator powers $a_e$ has  superficial degree of convergence
\begin{align}\label{def:sdd}
\omega_\Graph &:= \sum_{e\in E_\Graph} a_e - \loopnumber \frac D 2,
\end{align}
and its integral is proportional to $(p^2)^{-\omega_\Graph}$. This means that the operation $B_+[\Graph]$ of inserting a subdiagram $\Graph$ into an edge $e$, where $k_e$ is the edge momentum,  is the same as replacing the exponent of the propagator of $e$ by an appropriate non-integer value, 
\begin{align*}
\frac{1}{(k_e^2)} \mapsto \frac{1}{(k_e^2)^{1+\omega_\Graph}}.
\end{align*}
Consequently, we can compute \emph{all} of the Feynman diagrams if we can compute the kernel diagram for arbitrary powers of the propagator where the insertion happens. Let $k_1, \ldots, k_{\abs{E_\Graph}}$ be the edge momenta  of a kernel diagram $\Graph$, and assume we insert into $k_1$. The \emph{Mellin transform} of  $\Graph$ is defined by 
\begin{align}\label{def:Mellin_transform}
&F_\Graph(\epsilon, \rho) \\
&:= (p^2)^{\omega_\Graph} \int \frac{\d^{\loopnumber (D_0-2\epsilon)}}{(2\pi)^{\loopnumber (D_0-2\epsilon)}} \frac{1}{(k_1^2)^{1-\rho}} \frac{1}{(k_2^2)} \cdots \frac{1}{(k_{\abs{E_\Graph}})^2}.\nonumber 
\end{align}
This function depends on the integer spacetime dimension $D_0$, but we do not write this dependence since within one theory, $D_0$ is fixed. 
The prefactor in \cref{def:Mellin_transform} ensures that $F_\Graph$ is independent of the kinematic variable $p$. Since we assume that the kernel diagrams are superficially UV-divergent,   the Mellin transform has a simple pole of the form 
\begin{align}\label{Mellin_transform_pole}
F(\epsilon, \rho) = \frac{\period}{\loopnumber \epsilon -\rho}+ \text{regular terms},
\end{align}
where $\loopnumber$ is the loop number and $\period$ is the \emph{period} \cite{bloch_motives_2006,schnetz_quantum_2010,balduf_statistics_2023}.   

In our formalism, a Green's function is a formal power series in the coupling $\alpha$, a logarithmic kinematic variable $L= \ln ({p^2}/{\mu^2})$, and the dimensional regulator $\epsilon$. 
At each finite order in perturbation theory, we are therefore working with polynomials. 
The operator $B_+[\Graph]$, which in \cref{DSE_combinatorial} acts on Feynman diagrams by insertion, algebraically is a Hochschild 1-cocycle in the Hopf algebra of renormalization. Translated to polynomials, it gets replaced by a differential operator acting on the Mellin transform $F(0,\rho)$ of the kernel graph by  a polynomial $f(\alpha,L)$  \cite{panzer_renormalization_2015}:
\begin{align*}
B_+[f](L) &\mapsto  f(\alpha,\partial_\rho) e^{L \rho} F(0, \rho) \Big|_{\rho=0}. 
\end{align*}
The notation $f(\alpha, \partial_\rho)$ means that the parameter $L$ in $f(\alpha,L)$ is to be replaced by a differential operator which acts on all terms to the right. 
We see why   MOM renormalization conditions are preferred: In MOM, an expression is renormalized by subtracting the same expression at the kinematic renormalization point, which one can choose to be $L=0$ (i.e. $p^2=\mu^2$). Since the $L$-dependence of the cocycle is very simple, MOM conditions can be realized simply by replacing $e^{L \rho} \mapsto (e^{L \rho}-1)$.

To better understand the effects of minimal subtraction later on, it will prove beneficial to generalize \cref{DSE_combinatorial} to a \emph{non-linear} DSE, which allows an arbitrary, but fixed, insertion exponent $s\in \mb R$, 
\begin{align}\label{DSE_combinatorial_s}
	G_\ren &= 1+\alpha \left( 1-\renop \right) B_+ \left[ G_\ren^{1+s} \right] .
\end{align}
The linear DSE is recovered with $s=0$. 
The analytic version of this DSE, including the $\epsilon$-dependence, reads
\begin{align}\label{DSE_differential}
 &G_\ren(\alpha,\epsilon, L) \\
 &=1+ \alpha G_\ren^{1+s} (\alpha,\epsilon ,\partial_\rho) \left( e^{L(\rho-\loopnumber \epsilon)}-1 \right) F( \epsilon, \rho) \Big|_{\rho=0}. \nonumber 
\end{align}
Since the power series $G_\ren(\alpha,\epsilon,L)$ starts with 1, it can be computed recursively order by order from \cref{DSE_differential}.   The parenthesis $e^{L(\rho-\loopnumber \epsilon)}-1=L(\rho-\loopnumber \epsilon) + \mc O(\rho^2, \epsilon^2)$ cancels the simple pole of $F(\epsilon,\rho)$ from \cref{Mellin_transform_pole}, and the right hand side is regular at $\epsilon=0$, as it should be for a renormalized Green's function. To compute the solution in MOM  at the physical dimension $\epsilon=0$, it is sufficient to work with $F(0, \rho)$.  

The Green's function in MOM is unity at $L=0$ (recall that $G_\ren$ denotes the projection onto suitable tensors, such that the tree level term is indeed 1), it therefore has a power series expansion of the form 
\begin{align}\label{G_expansion}
	G_\ren(\alpha,\epsilon,L) &= 1+\sum_{j=1}^\infty \gamma_j(\alpha,\epsilon) L^j.
\end{align}
Here, $\gamma_1(\alpha,\epsilon)=:\gamma(\alpha,\epsilon)$ is the anomalous dimension. In our setting, where there are no other Green's functions, the beta function of the theory is $\beta(\alpha, \epsilon)= s \alpha \gamma(\alpha,\epsilon)- \alpha \epsilon$, where $s$ is the exponent from \cref{DSE_combinatorial_s}. The extra factor $\alpha \epsilon$ represents the explicit scale dependence of the coupling constant in a non-integer spacetime dimension, to be discussed in \cref{sec:scheme-dependent}. The Callan-Symanzik  equation  \cite{callan_broken_1970,symanzik_small_1970},
\begin{align}\label{Callan_Symanzik_equation_1}
\partial_L G_\ren(\alpha,\epsilon,L) &= \Big( \gamma(\alpha, \epsilon) + \beta(\alpha,\epsilon) \partial_\alpha \Big) G_\ren(\alpha, \epsilon, L),
\end{align}
implies that all higher functions $\gamma_j(\alpha, \epsilon)$ in \cref{G_expansion} are determined from $\gamma_1$ according to 
\begin{align}\label{gamma_recurrence}
  \gamma_j(\alpha,\epsilon) &= \frac 1 j\Big( \gamma(\alpha,\epsilon)  + \big( s \alpha \gamma(\alpha,\epsilon) - \epsilon\alpha  \big) \partial_\alpha \Big) \gamma_{j-1}(\alpha, \epsilon).
\end{align}
Hence, solving the DSE is equivalent to determining $\gamma_1(\alpha, \epsilon)$, which at the same time is the anomalous dimension and the seed for the recurrence for $\gamma_j$. Inserting \cref{G_expansion} into the DSE \cref{DSE_differential}, and expanding with respect to $L$, produces a pseudo-differential equation for $\gamma(\alpha, \epsilon)$ \cite{balduf_dyson_2024}: 
\begin{align}\label{MOM_differential_general}
\frac{1}{  F(\epsilon, \rho+\loopnumber \epsilon)} \Big|_{\rho \mapsto \gamma(\alpha, \epsilon)+ \normalsize(s \alpha \gamma(\alpha, \epsilon) - \epsilon \alpha \normalsize)\partial_\alpha }   &= \alpha .
\end{align}
The differential operator on the LHS appears to \enquote{act on nothing}, in fact it acts on $\gamma(\alpha, \epsilon)$ itself. One may equivalently include another factor $\frac 1 {\rho}$, this  makes the argument  explicit:
\begin{align*}
	\frac{1}{  \rho \cdot  F(\epsilon, \rho+\loopnumber \epsilon)} \Big|_{\rho \mapsto \gamma(\alpha, \epsilon)+  (s \alpha \gamma(\alpha, \epsilon) - \epsilon \alpha  )\partial_\alpha } \gamma(\alpha, \epsilon)  &= \alpha .
\end{align*}
  For later use, we introduce the series expansion 
\begin{align}\label{mellin_T_series}
\frac{1}{ F(\epsilon, \rho+ \loopnumber\epsilon)} &=: T_0(\rho) \cdot  \Big( 1+  \sum_{j=1}^\infty  \epsilon^j T_j(\rho) \Big) , 
\end{align}
so that the limit $\epsilon \rightarrow 0$ involves only $T_0(\rho) = \frac{1}{ F(0 , \rho)}$ and produces the pseudo-differential equation 
\begin{align}\label{ODE}
T_0(\gamma + s \alpha \gamma \partial_\alpha) &=\alpha.
\end{align}
The discovery of this differential equation for the   cases of Yukawa and $\phi^3$ theory by Broadhurst and Kreimer in \cite{broadhurst_exact_2001}, and its exact solution, was the starting point for the development of the present formalism. It draws its power from the fact that in MOM, one works with power series in only one parameter, $\alpha$. Versions of \cref{DSE_differential,ODE} have allowed for symbolic and numerical calculations to very high loop order \cite{bellon_efficient_2010,bellon_renormalization_2008,bellon_approximate_2010} and the study of their asymptotic and resurgent features \cite{bellon_alien_2017,bellon_schwinger_2015,borinsky_nonperturbative_2020,borinsky_semiclassical_2021,bellon_resurgent_2021,bellon_ward_2021,borinsky_taming_2022,borinsky_treetubings_2024}, where the series coefficients have a combinatorial interpretation in terms of chord diagrams \cite{marie_chord_2013,courtiel_nexttok_2020,hihn_generalized_2019,courtiel_terminal_2017} and, more recently, tubings of rooted trees \cite{balduf_tubings_2024,borinsky_treetubings_2024,olson-harris_algebraic_2025}. These references also contain generalizations to multiple kernels, to multiple insertion places, and to systems of coupled  DSEs.

\bigskip 

Notice that the entire procedure never requires us to explicitly work with counterterms. We are free to include the $\epsilon$-dependence of renormalized quantities, but we can also work at $\epsilon=0$ throughout. Nevertheless, the formalism is entirely consistent with ordinary multiplicative renormalization, namely with 
\begin{align}\label{Gren_Z}
G_\ren(\alpha,\epsilon, L) &= Z_2 \cdot G_0 \big( Z_\alpha \mu^{2\epsilon}\cdot \alpha, ~\epsilon, ~L \big) .
\end{align}
The first argument of the bare Green's function $G_0$, $Z_\alpha \mu^{2\epsilon}\alpha=:\alpha_0$, is the bare coupling. It has mass dimension $2\epsilon$, whereas the renormalized coupling $\alpha$ is dimensionless.
The counterterms are related to the renormalization group functions by the Gross t'Hooft relations \cite{thooft_dimensional_1973,gross_applications_1981}, and the equation $\beta(\alpha, \epsilon) = s\alpha \gamma(\alpha, \epsilon)-\alpha \epsilon $ is equivalent to $Z_\alpha = Z_2^s$, concretely
\begin{align}\label{Za_beta}
\beta(\alpha,\epsilon) &= \frac{-\epsilon \alpha}{1+ \alpha \partial_\alpha \ln (Z_\alpha(\alpha,\epsilon))}, \\
\gamma(\alpha, \epsilon)  &= - \beta(\alpha, \epsilon) \partial_\alpha \big(\ln Z_2(\alpha, \epsilon) \big), \nonumber \\
Z_2 &= \exp \left( -\int_0^\alpha \frac{\d u}{u} \frac{\gamma(u,\epsilon)}{s\gamma(u,\epsilon)-\epsilon} \right) , \nonumber  \\
 Z_\alpha &=  \exp \left( -\int_0^\alpha \frac{\d u}{u} \frac{s\gamma(u,\epsilon)}{s\gamma(u,\epsilon)-\epsilon} \right)= Z_2^s .\nonumber 
\end{align}

\subsection{The linear DSE in MOM}\label{sec:linear_MOM}

If the DSE is linear, that is, $s=0$, \cref{MOM_differential_general} becomes
\begin{align}\label{linear_ODE_epsilon}
	\frac{1}{\rho F(\epsilon, \rho+\loopnumber\epsilon)} \Big|_{\rho \mapsto \gamma- \epsilon \alpha  \partial_\alpha } \gamma(\alpha,\epsilon) &= \alpha .
\end{align}
At this point it becomes clear why we generalized to the non-linear DSE \cref{DSE_combinatorial_s}: As long as $\epsilon \neq 0$, even the linear DSE still gives rise to a differential equation, similar to the non-linear DSE at $\epsilon=0$ (\cref{ODE}). However, in \cref{linear_ODE_epsilon} the   $\epsilon$- and $\alpha$-dependence of $\gamma(\alpha,\epsilon)$ is essentially decoupled: We can first set $\epsilon=0$ to obtain the algebraic equation 
\begin{align}\label{MOM_algebraic_equation}
 T_0(\gamma(\alpha))&= \alpha. 
\end{align}
Its solution $\gamma(\alpha):= \gamma(\alpha,0)$ is the exact anomalous dimension of the linear DSE at the physical dimension $D_0$.   We can then construct a power-series expansion in $\epsilon$ order by order from \cref{linear_ODE_epsilon}.  
Knowing $\gamma(\alpha,\epsilon)$, we get the full Green's function from the recurrence \cref{gamma_recurrence}. Equivalently, we solve the Callan-Symanzik equation 
\begin{align}\label{linear_Callan_Symanzik}
\partial_L G_\ren(\alpha,\epsilon, L) &= \left( \gamma(\alpha,\epsilon) - \epsilon \alpha \partial_\alpha \right) G_\ren(\alpha,\epsilon,L)
\end{align}
by separation of variables, imposing MOM conditions  $G_\ren(\alpha,\epsilon,0)=1$. This leads to  
\begin{align}\label{MOM_integral_solution}
G_\ren(\alpha,\epsilon, L) &= \exp \left( \int_{\alpha e^{-\epsilon L}}^\alpha \d u \; \frac{\gamma(u, \epsilon)}{u \epsilon} \right) . 
\end{align}
Both approaches produce, in the physical limit $\epsilon \rightarrow 0$,  the scaling solution
\begin{align}\label{MOM_scaling_solution}
G_\ren(\alpha,0, L) &= \exp \big( L \cdot  \gamma(\alpha)    \big)  = \left( \frac{p^2}{\mu^2} \right) ^{\gamma(\alpha)}.
\end{align}

\section{Minimal subtraction}

\subsection{Scheme-dependent renormalization group} \label{sec:scheme-dependent}

We now consider an arbitrary non-kinematic renormalization scheme, which will later be specialized to MS.  Our strategy is to relate such a scheme to kinematic renormalization, in order to be able to use the machinery of \cref{sec:MOM}. Concretely, we are constructing a power series $\delta(\alpha, \epsilon)$ such that the MS Green's function equals a MOM Green's function, but with renormalization point $L=-\delta$ instead of $L=0$. A crucial insight from \cref{sec:MOM,sec:linear_MOM} is that for $\epsilon \neq 0$, the theory has a non-vanishing beta function, even in the linear case, arising from the fact that we have tacitly absorbed a scale $\mu^{2\epsilon}$ into $\alpha=\alpha_0 \mu^{-2\epsilon}Z_\alpha^{-1}$ (\cref{Gren_Z}). An analogous phenomenon occurs when we express a MS solution through a shifted MOM solution: In MS, we are working with an expansion parameter $\bar \alpha$ which is related to the MOM expansion parameter $\alpha$ by $\bar \alpha (\alpha)= \alpha e^{-\epsilon \delta(\alpha, \epsilon)}$. Another way to see this is that the logarithmic scale $\delta(\alpha,\epsilon)= \ln \frac{\Delta^2}{\mu^2}$ amounts to  some non-logarithmic momentum scale $\Delta(\alpha, \epsilon)$, and the MS-coupling $\bar \alpha$ is expressed relative to \emph{that} scale, 
\begin{align*}
\bar \alpha &= \alpha e^{-\epsilon \delta} = \alpha_0 Z_\alpha^{-1}\mu^{-2\epsilon} \left( \frac{\Delta^2}{\mu^2} \right) ^{-\epsilon} = \alpha_0 Z_\alpha^{-1}\Delta^{-2 \epsilon}. 
\end{align*}
The distinction $\alpha \neq \bar \alpha$ only appears since we explicitly relate MS to MOM. If one works in MS throughout, one would always be using $\bar \alpha$ (and simply call it $\alpha$).

We thus demand that the MS-renormalized $\bar G_{ \ren}$ should be related to the MOM-renormalized $G_\ren$ by
\begin{align}\label{G_shifted}
\bar G_{ \ren}(\bar \alpha(\alpha), \epsilon, L) &= G_\ren\big(\alpha, \epsilon , L+\delta(\alpha, \epsilon)\big). 
\end{align} 
This relation constitutes the definition of $\delta(\alpha, \epsilon)$. The fact that it is possible to find such $\delta(\alpha, \epsilon)$ is obvious in perturbation theory and discussed at length in \cite{balduf_dyson_2023,balduf_dyson_2024}: At every new order in $\alpha$, the renormalized Green's function is a polynomial in $\alpha$ and $L$, and one can introduce an arbitrary finite shift of its value by adding an offset, of the same order in $\alpha$, to $L$. In fact, this mechanism is completely analogous to the mechanism that allows the order-by-order construction of counterterms, just that it involves only finite shifts. Conversely, a choice of $\delta(\alpha, \epsilon)$ amounts to a choice of perturbative renormalization scheme, and the kinematic schemes are exactly those where $\delta$ is independent of $\alpha$ and $\epsilon$. 

For the MS solution, too, we can write a generic expansion of the form \cref{G_expansion}, but this time there is an additional non-trivial function $\bar \gamma_0(\bar \alpha, \epsilon)$:  
\begin{align}\label{G_expansion_MS}
	\bar G_{ \ren}(\bar \alpha,\epsilon,L) &= \sum_{j=0}^\infty \bar \gamma_j(\bar \alpha,\epsilon) L^j.
\end{align}
With our definitions, the Callan-Symanzik equation (\ref{Callan_Symanzik_equation_1}) takes the same form in all schemes, 
\begin{align}\label{Callan_symanzik_equation}
	\partial_L \bar G_{  \ren}(\bar \alpha,\epsilon,L) &= \Big( \bar \gamma(\bar \alpha, \epsilon) + \bar \beta(\bar \alpha,\epsilon) \partial_{\bar \alpha} \Big) \bar G_{ \ren}(\bar \alpha, \epsilon, L).
\end{align}
Since we are still working with a single DSE (\cref{DSE_combinatorial_s}), the renormalization group functions are still related by
\begin{align}\label{MS_beta_gamma}
\bar \beta(\alpha, \epsilon) &= s\bar \alpha \bar \gamma(\bar \alpha, \epsilon) - \bar \alpha \epsilon. 
\end{align}
Recall that in MOM, the first expansion function $\gamma_1(\alpha, \epsilon)$ coincides with the anomalous dimension $\gamma(\alpha, \epsilon)$. This is not true in general renormalization schemes. Instead, from inserting \cref{G_expansion_MS} into \cref{Callan_symanzik_equation}, we have for all schemes
\begin{align}\label{anomalous_dimension}
	\bar \gamma(\bar \alpha,\epsilon) &=  \frac{\bar \gamma_1(\bar \alpha,\epsilon) - \epsilon \bar \alpha \partial_{\bar \alpha} \bar \gamma_0(\bar \alpha,\epsilon) }{\left( 1+ s \bar\alpha \partial_{\bar \alpha} \right) \bar \gamma_0(\bar \alpha,\epsilon)}.
\end{align}

 We insert both series expansions (\cref{G_expansion,G_expansion_MS}) into \cref{G_shifted} to obtain
\begin{align}\label{gamma_delta_rescaling}
\bar \gamma_k(\bar \alpha(\alpha), \epsilon) &= \bar \gamma_k (\alpha e^{-\epsilon \delta(\alpha,\epsilon)}, \epsilon)  \\
&= \sum_{j=k}^\infty \binom j k \gamma_j(\alpha,\epsilon) \delta^{j-k} (\alpha,\epsilon).\nonumber
\end{align}
This reveals that of the three functions $\left \lbrace \bar \gamma(\bar \alpha, \epsilon), \gamma(\alpha, \epsilon), \delta(\alpha, \epsilon) \right \rbrace $, only two are independent. 
In \cref{gamma_delta_rescaling}, it is natural to use  the variable $\alpha$ of the MOM solution since we have defined $\delta=\delta(\alpha, \epsilon)$ as a function of $\alpha$. However, we can revert these series:
\begin{align}\label{delta_bardelta}
	\bar \delta(\bar \alpha, \epsilon)&:= \delta(\alpha(\bar \alpha, \epsilon), \epsilon), \quad \bar \alpha = \alpha e^{-\epsilon \delta(\alpha, \epsilon)} \\
	\Leftrightarrow \quad \alpha &= \bar \alpha e^{+\epsilon \delta(\alpha(\bar \alpha, \epsilon), \epsilon)} = \bar \alpha e^{\epsilon \bar \delta(\bar \alpha, \epsilon)}.\nonumber
\end{align}
From now on we leave out the arguments; $\bar \delta$ is a function of $\bar \alpha$ and $\delta$ is a function of $\alpha$.  The chain rule implies 
\begin{align}\label{delta_deltabar_chain_rule}
	\frac{\partial \bar \alpha}{\partial \alpha}&= \frac{\bar \alpha}{\alpha}-\epsilon \bar \alpha \partial_\alpha \delta  \\
	\Rightarrow \quad \alpha \partial_\alpha \delta&= \frac{\bar \alpha \partial_{\bar \alpha} \bar \delta}{1+ \epsilon \bar \alpha \partial_{\bar \alpha} \bar \delta}, \qquad \bar \alpha \partial_{\bar \alpha} \bar \delta = \frac{\alpha \partial_\alpha \delta}{1-\epsilon \alpha \partial_\alpha \delta}.\nonumber
\end{align}

\begin{theorem}\label{lem:MS_anomalous_dimension}
	With the shift $\delta$ resp. $\bar \delta$ from   \cref{delta_bardelta}, the shifted anomalous dimension $\bar \gamma(\bar \alpha, \epsilon)$ is related to the MOM anomalous dimension $\gamma(\alpha, \epsilon)$ via
	\begin{align*}
		\bar \gamma\big(\bar \alpha(\alpha), \epsilon \big) &= \frac{  \gamma(\alpha,\epsilon)}{ 1+ \left(   s \alpha \gamma(\alpha,\epsilon) - \epsilon\alpha   \right)  \partial_\alpha \delta }, \\
	 \textnormal{equivalently}\quad \bar \gamma(\bar \alpha, \epsilon)&= \frac{\gamma\big(\alpha(\bar \alpha), \epsilon\big) \big( 1+ \epsilon \bar \alpha \partial_{\bar \alpha} \bar \delta \big)} {1+s \gamma\big(\alpha(\bar \alpha), \epsilon\big) \bar \alpha \partial_{\bar \alpha} \bar \delta }.
	\end{align*} 
\end{theorem}
\begin{proof}
	\Cref{gamma_recurrence} implies that 
	\begin{align*} 
		\frac{(j+1)	\gamma_{j+1}(\alpha,\epsilon) - \gamma(\alpha,\epsilon) \gamma_j(\alpha,\epsilon)}{  s \alpha \gamma(\alpha,\epsilon) - \epsilon\alpha   } &=      \partial_\alpha  \gamma_j(\alpha, \epsilon).
	\end{align*}
Derive \cref{gamma_delta_rescaling} with respect to $\alpha$ and insert the previous equation.
	\begin{align*}
		& \frac{\partial\bar \alpha}{\partial \alpha} \partial_{\bar \alpha} \bar \gamma_k(\bar \alpha,\epsilon) 
		  = \sum_{j=k}^\infty \binom j k \partial_\alpha \gamma_j(\alpha,\epsilon) \delta^{j-k} (\alpha,\epsilon)\\
		 &\qquad  + \sum_{j=k}^\infty \binom j k (j-k) \gamma_j(\alpha,\epsilon) \delta^{j-k-1} (\alpha,\epsilon) \partial_\alpha \delta\\
		&= \sum_{j=k}^\infty \binom j k \frac{(j+1)	\gamma_{j+1}(\alpha,\epsilon) - \gamma(\alpha,\epsilon) \gamma_j(\alpha,\epsilon)}{  s \alpha \gamma(\alpha,\epsilon) - \epsilon\alpha   }\delta^{j-k} (\alpha,\epsilon) \\
		&\qquad + \sum_{j=k}^\infty \binom j {k+1} (k+1) \gamma_j(\alpha,\epsilon) \delta^{j-k-1} (\alpha,\epsilon) \partial_\alpha \delta\\
		&=  \frac{ (k+1) \bar \gamma_{k+1} (\bar \alpha,\epsilon) }{  s \alpha \gamma(\alpha,\epsilon) - \epsilon\alpha   }\  -\frac{  \gamma(\alpha,\epsilon) \bar \gamma_k(\bar \alpha,\epsilon)}{  s \alpha \gamma(\alpha,\epsilon) - \epsilon\alpha   }  \\
		&\qquad  + \partial_\alpha \delta\cdot (k+1) \bar \gamma_{k+1}(u, \epsilon) .
	\end{align*}
	Solve the equation for $\bar \gamma_{k+1}$. The  chain rule of \cref{delta_deltabar_chain_rule} cancels a denominator, and ensures that the resulting equation,
	\begin{align*}
		&(k+1) \bar \gamma_{k+1} (\bar \alpha,\epsilon) \\
		&= \left(\frac{ s \bar \alpha \gamma(\alpha,\epsilon) }{ 1+ \left(   s \alpha \gamma(\alpha,\epsilon) - \epsilon\alpha   \right)  \partial_\alpha \delta  }  - \epsilon \bar \alpha   \right)\partial_{\bar \alpha} \bar \gamma_k(\bar \alpha,\epsilon)  \\
		&\qquad 		+\frac{  \gamma(\alpha,\epsilon)  \bar \gamma_k(\bar \alpha,\epsilon) 	 }{ 1+ \left(   s \alpha \gamma(\alpha,\epsilon) - \epsilon\alpha   \right)  \partial_\alpha \delta  },
	\end{align*}
	has precisely the expected form
	\begin{align*}
		&(k+1) \bar \gamma_{k+1}(\bar \alpha,\epsilon) \\
		&= \Big(s\bar \alpha \bar \gamma(\alpha,\epsilon) - \epsilon \bar \alpha \Big) \partial_{\bar \alpha} \gamma_k(\bar \alpha,\epsilon) + \bar \gamma(\bar \alpha,\epsilon) \bar \gamma_k(\bar \alpha,\epsilon). 
	\end{align*}
	The second formula then follows from \cref{delta_deltabar_chain_rule}.
\end{proof}
The special case $\epsilon=0$ of \cref{lem:MS_anomalous_dimension}  had been given already in \cite{balduf_dyson_2023}.

\subsection{Beta function in minimal subtraction}
The construction in \cref{sec:scheme-dependent} has been for an arbitrary renormalization scheme obtained through a shift $\delta(\alpha,\epsilon)$. We now specialize to MS.

The defining property of MS is that the counterterms only include pole terms in $\epsilon$. We need to translate this to a statement about the renormalization group functions, because our formalism does not involve the counterterms explicitly. To this end, we use that the counterterm is related to the beta function by \cref{Za_beta}. Introduce the function $B(\bar \alpha,\epsilon):= \bar \beta(\bar \alpha,\epsilon) + \bar \alpha \epsilon$, then
\begin{align*}
Z_\alpha &= \exp \left(  \int_0^{\bar \alpha} \frac{\d u}{u} \frac{B(u,\epsilon)}{u\epsilon-B(u,\epsilon)} \right)  \\
&= \exp \left( \int_0^\alpha \frac{\d u}{u} \frac{B(u,\epsilon)}{u\epsilon} \sum_{j=0}^\infty \left( \frac{B(u,\epsilon)}{u \epsilon} \right) ^j  \right) .
\end{align*}
The right hand side should be viewed as a power series in $\bar \alpha$. In MS, it is required to consist of pole terms in $\epsilon$ exclusively. This implies that $\frac{B(u,\epsilon)}{u \epsilon}$ consists  of poles only. On the other hand, the beta function itself is regular in $\epsilon$, therefore $B(\bar \alpha,\epsilon)$ does not contain poles in $\epsilon$. The only remaining possibility is that $B(\bar \alpha,\epsilon)$ does not depend on $\epsilon$ at all. We obtain a well-known alternative definition of the MS scheme: The beta function $\bar \beta(\bar\alpha, \epsilon)$ in MS depends on $\epsilon$ only through a single term, 
\begin{align}\label{MS_beta}
\bar \beta(\bar \alpha,\epsilon) &= \bar \beta(\bar \alpha) - \bar \alpha \epsilon .
\end{align}
One can repeat the same argument for the counterterm $Z_2$ and its relation to the anomalous dimension $\gamma(\alpha,\epsilon)$, or one uses \cref{MS_beta_gamma}.
In either case, one finds that in MS the anomalous dimension $\bar \gamma(\bar \alpha,\epsilon) = \bar \gamma(\bar \alpha)$ is entirely independent of $\epsilon$. 
This restricts the dependence of the expansion functions $\gamma_j(\alpha,\epsilon)$ of \cref{G_expansion_MS} on $\epsilon$, but it does not imply that they, too, are independent. Namely \cref{gamma_recurrence} reads
\begin{align}\label{gamma_recurrence_MS}
 \bar\gamma_j(\bar \alpha,\epsilon) &= \frac 1 j\Big(\bar \gamma(\bar \alpha) + \big(s \bar \alpha \bar \gamma(\bar \alpha) - \epsilon \bar \alpha \big) \partial_{\bar \alpha} \Big)\bar \gamma_{j-1}(\bar \alpha, \epsilon) 
\end{align}

\subsection{The linear DSE in MS}

In the linear case, $s=0$, the second formula in \cref{lem:MS_anomalous_dimension} simplifies and the anomalous dimensions of MOM and MS are related via
\begin{align}\label{MS_anomalous_dimension_linear}
\bar \gamma(\bar \alpha, \epsilon)= \gamma(\alpha(\bar \alpha), \epsilon) \cdot \big( 1+ \epsilon \bar \alpha \partial_{\bar \alpha} \bar \delta(\bar \alpha, \epsilon) \big) .
\end{align}
The anomalous dimension in MS, as a function of $\bar \alpha$, is independent of $\epsilon$. Consequently, it coincides with the limit $\epsilon \rightarrow 0$ of the anomalous dimension in MOM, using $\bar \alpha= \alpha + \mc O(\epsilon)$:
\begin{align}\label{MS_linear_anomalous_dimension}
\bar \gamma(\bar \alpha, \epsilon) = \bar \gamma(\bar \alpha) &= \gamma(\bar \alpha, 0)=\gamma( \alpha).
\end{align}
All our definitions have been engineered such that the Callan-Symanzik equation (\cref{linear_Callan_Symanzik}) holds, in exactly the same form, in all schemes.   Consequently, the solution of this equation in MS  has the same form as in MOM, namely \cref{MOM_integral_solution}, with $\gamma$ from \cref{MS_linear_anomalous_dimension}. The only difference is that in MS, we do not have the boundary condition at $L=0$, and therefore, the solution is multiplied by the undetermined factor $\bar \gamma_0(\bar \alpha, \epsilon)$ of \cref{G_expansion_MS}: 
\begin{align}\label{G_MS_integral}
	\bar G_{  \ren}(\bar \alpha, \epsilon, L) &= \bar \gamma_0(\bar \alpha, \epsilon)\cdot \exp \left( \int_{\bar \alpha e^{-\epsilon L}}^{\bar \alpha} \d   u \frac{\bar \gamma(  u)}{  u \epsilon } \right) \nonumber \\
	&=\bar \gamma_0(\bar \alpha, \epsilon) \cdot G_\ren(\bar \alpha, \epsilon, L).
\end{align}
Setting $L=0$ in \cref{G_shifted}, one obtains
\begin{align}\label{gamma0_delta}
	\bar \gamma_0(\bar \alpha (\alpha), \epsilon) &= G_\ren(\alpha, \epsilon, \delta(\alpha, \epsilon)),
\end{align}
We remark that this construction is consistent, in the sense that one can insert   \cref{MS_anomalous_dimension_linear} into the integral of \cref{G_MS_integral},  do a  change of variables $\bar \alpha = \alpha e^{-\epsilon \delta}$  with \cref{delta_deltabar_chain_rule}, and identify the integral \cref{MOM_integral_solution}, to recover the MOM solution written in terms of $\alpha$:
\begin{align*}
	\exp \left( \int_{\bar \alpha e^{-\epsilon L}}^{\bar \alpha} \d \bar u \frac{\bar \gamma(\bar u)}{\bar u \epsilon } \right) 
	&= \frac{G_\ren \big(\alpha, \epsilon, L+\delta(\alpha,\epsilon) \big)}{G_\ren(\alpha, \epsilon, \delta)}.
\end{align*}

Another useful perspective on the quantity $\bar \gamma_0$ is to view an overall multiplicative scaling of $G_\ren$ as a scaling of the counterterm $Z_2$ according to \cref{Gren_Z}, namely $\bar \gamma_0(\bar \alpha, \epsilon) Z_2(\bar \alpha, \epsilon)  = \bar Z_2(\bar \alpha, \epsilon)$.
By \cref{Za_beta}, the counterterm $Z_2$ determines the anomalous dimension,
\begin{align}\label{MS_MOM_gamma_gamma0}
\bar \gamma(\bar \alpha, \epsilon) &= \bar \alpha \epsilon \partial_\alpha \ln \big(\bar Z_2(\bar \alpha, \epsilon) \big)  \nonumber \\
&=  \epsilon \bar \alpha \partial_{\bar \alpha} \ln \big(\bar \gamma_0(\bar \alpha, \epsilon)  \big)  + \gamma(\bar \alpha, \epsilon).
\end{align}
This formula relates the MOM and MS anomalous dimensions similar to \cref{MS_anomalous_dimension_linear}, but in terms of $\bar \gamma_0$ instead of $\delta$. It has the additional benefit that it does not involve a change of variables $\alpha \leftrightarrow \bar \alpha$. We can integrate it, using   \cref{MS_linear_anomalous_dimension}, to compute $\bar \gamma_0$ explicitly:
\begin{align}\label{gamma0_integral}
\bar \gamma_0(\bar \alpha, \epsilon) &= \exp \left( -\int_0^{\bar \alpha} \frac{\d  u}{  u} \frac{ \gamma(  u , \epsilon)-\gamma(  u)}{ \epsilon}\right).
\end{align}

\subsection{Physical spacetime dimension}
We still consider the linear DSE, $s=0$, in MS. 
For the physical limit $\epsilon=0$, where $\alpha=\bar \alpha$, the situation simplifies further. The MOM Green's function is then the scaling solution of \cref{MOM_scaling_solution},  $G_\ren(\alpha,0,  \delta) = e^{\delta(\alpha) \cdot \gamma(\alpha)}$. Consequently,  \cref{gamma0_delta} implies 
 $\bar \gamma_0(  \alpha )  = e^{\delta(\alpha ) \cdot \gamma(\alpha)}$.
\begin{theorem}\label{thm:DSE_solution_linear}
The solution of a linear DSE, in the physical dimension $\epsilon=0$, is  given by
\begin{align*}
	\bar G_{ \ren}(\alpha, L) &= \bar \gamma_0(\alpha) e^{L \gamma(\alpha)} = \exp \Big(\big(L+\delta(\alpha)\big)\gamma(\alpha)\Big),
\end{align*}
where $\gamma(\alpha)$ is the anomalous dimension in MOM,  computed from \cref{MOM_algebraic_equation}, and 
\begin{align*} 
	\bar \gamma_0(\bar \alpha, 0) &= \exp \left( -\int_0^{\bar \alpha} \frac{\d  u}{  u} g(  u ) \right), \\
	\delta(\alpha, 0) &= \frac{\ln \bar \gamma_0(  \alpha )}{\gamma(\alpha)}, ~\text{ using } \bar \alpha = \alpha + \mc O(\epsilon)\nonumber.
\end{align*}
The function $g(\alpha)$ will be determined in \cref{lem:g_Mellin}.
\end{theorem} 

The $\epsilon$-independent function  $\bar \gamma_0(\bar \alpha)=\bar \gamma_0(\bar \alpha, 0)$ is, by \cref{gamma0_integral}, entirely determined by the order $\epsilon^1$-term of the MOM anomalous dimension, which we call  $g(\alpha):=[\epsilon^1]\gamma(\alpha, \epsilon)$. This, in turn, can be computed from the $\epsilon$-dependent ODE-version of the Dyson-Schwinger equation, \cref{linear_ODE_epsilon}.

\begin{theorem}\label{lem:g_Mellin}
	With the expansion functions of \cref{mellin_T_series}, $\frac{1}{F(\epsilon, \rho+\epsilon)}= T_0(\rho) + \epsilon T_0 (\rho) T_1(\rho) + \ldots$, the order $[\epsilon^1]$ of the MOM anomalous dimension $\gamma(\alpha,\epsilon)$ of a linear DSE is given by 
	\begin{align*}
		g(\alpha) &= \alpha \frac{\frac 12 \partial^2_\rho T_0 \big|_{\rho=\gamma(\alpha)} \cdot   \partial_\alpha \gamma(\alpha)  - T_1(\gamma(\alpha))}{\partial_\rho T_0\big|_{\rho=\gamma(\alpha)}}.
	\end{align*}
\end{theorem}
\begin{proof}
	In the proof, we write $\gamma$ for $\gamma(\alpha)=\gamma(\alpha,0)$ and $g$ for $g(\alpha)$. 
	The  ODE \cref{linear_ODE_epsilon} with \cref{mellin_T_series} is
	\begin{align}\label{ODE_expansion_epsilon} 
		\Big(T_0(\rho) + \epsilon T_0(\rho)T_1(\rho)  + \ldots \Big)_{\rho \mapsto \gamma(\alpha, \epsilon)-\epsilon \alpha \partial_\alpha }=\alpha. 
	\end{align}
	We  take the order $[\epsilon^1]$ of this equation, the $\ldots$ terms are of higher order and do not contribute. The second summand is already at order $\epsilon^1$, hence we merely replace $\rho \mapsto \gamma $ and use \cref{MOM_algebraic_equation}, $T_0(\gamma)=\alpha$, and the summand becomes
	\begin{align*}
		\big( T_0(\rho)T_1(\rho) \big) _{\rho \mapsto \gamma } &= \alpha T_1(\gamma ). 
	\end{align*}
	For the first summand in \cref{ODE_expansion_epsilon}, it is clear that we need at most the linear order in $\epsilon$ of the argument, $\gamma(\alpha, \epsilon)-\epsilon \alpha \partial_\alpha = \gamma + \epsilon(g -\alpha \partial_\alpha) + \mc O(\epsilon^2)$. We write $T_0(\rho)= \sum_{j=1}^\infty  t_j \rho^j$ and consider a fixed order $j$. 
	\begin{align*}
		\left[ \epsilon^1 \right]  \left( \gamma  + \epsilon(g -\alpha \partial_\alpha)  \right) ^j =  \sum_{k=0}^{j-1} \gamma^k (g-\alpha \partial_\alpha) \gamma^{j-k-1}.
	\end{align*}
	The factor $g$ can be pulled out, giving $\gamma^{j-k-1}\cdot g$. With the chain rule, the derivative term becomes 
	\begin{align*}
	 \sum_{k=0}^{j-1} \gamma^k  (j-k-1)\gamma^{j-k-2} \partial_\alpha \gamma = \frac{j(j-1)}{2}\gamma^{j-2} \cdot  \partial_\alpha \gamma.
	\end{align*}
	In both cases, the summand can be interpreted as a summand of a derivative of $T_0$, therefore
	\begin{align*}
		\left[ \epsilon^1 \right] T_0(\ldots) &=  \left(  \partial_\rho T_0 \right) _{\rho \rightarrow \gamma} \cdot g - \frac \alpha 2\left( \partial^2 \rho T_0 \right) _{\rho \rightarrow \gamma} \cdot \partial_\alpha \gamma. 
	\end{align*}
	We see that in the order $[\epsilon^1]$ of \cref{ODE_expansion_epsilon}, the sought-after $g $ appears as a factor,
	\begin{align*}
		\left(  \partial_\rho T_0 \right) _{\rho \rightarrow \gamma} \cdot g - \frac \alpha 2 \left( \partial^2 \rho T_0 \right) _{\rho \rightarrow \gamma} \cdot \partial_\alpha \gamma + \alpha T_1(\gamma)=0.
	\end{align*}
\end{proof}

As long as we are able to compute the Mellin transform $T_0, T_1$ in closed form, we obtain $g(\alpha)$ in closed form. Furthermore, notice that 
\begin{align*}
\partial_\rho^2 T_0 \big|_{\rho=\gamma(\alpha)} \partial_\alpha \gamma(\alpha) = \partial_\alpha \big(\partial_\rho T_0\big|_{\rho=\gamma(\alpha)}\big)
\end{align*}
and therefore the integral
\begin{align}\label{integral_simplification}
&\int_0^{\bar \alpha}\frac{\d u}{u}\; u \frac{\frac 12 \partial_\rho^2 T_0\big|_{\rho=\gamma(u)} \partial_u \gamma(u)}{\partial_\rho T_0 \big|_{\rho=\gamma(u)}}\\
&= \int_0^{\bar \alpha}\frac 12 \d u \; \partial_u \ln \left( \partial_\rho T_0 \big|_{\rho=\gamma(u)}  \right)  = \frac 12 \ln \left( \partial_\rho T_0 \big|_{\rho=\gamma(\bar \alpha)} \nonumber  \right) 
\end{align}
can always be solved analytically. 
The only potentially non-trivial integration in the computation of $\bar \gamma_0$ in \cref{thm:DSE_solution_linear} is that of the second summand, $\frac{T_1}{\partial_\rho T_0}$ in \cref{lem:g_Mellin}, so that we have a good chance of finding a closed-form solution for $\bar \gamma_0$ provided we know $T_1$ and $T_0$ in closed form, as claimed in the abstract. In all examples considered below, the integration can be done analytically.

\section{Examples}\label{sec:examples}

In the remainder of the article, we demonstrate how to use the formalism for some examples of linear DSEs. As will be explained around \cref{alpha_msbar}, we will actually be using the MS-bar scheme in order to streamline the notation.   All computations are implemented in a Mathematica notebook that is available as an attachment to the electronic version of the article and from the author's website 
\footnote{\url{https://paulbalduf.com/research}}.
To keep the paper short, we restrict ourselves to rather simple examples.
Recall that the Green's function $G_\ren$ is a projection onto a tree level tensor, therefore, it is a scalar quantity regardless of whether the fields are  scalars themselves. The formalism allows for the kernel diagram to have arbitrary loop number, however, since we are considering only one kernel, it is physically sensible to choose the one that has lowest loop number.

In $D=D_0-2\epsilon$ dimensions, with propagator powers $1$ and $1-\rho$, the  1-loop multiedge diagram of \cref{fig:B_phi3} has superficial degree of convergence  
\begin{align*}
\omega&= 2- \frac{D_0}{2}-\rho +\epsilon. 
\end{align*}
The multiedge  has two vertices, each of which has a Feynman rule $(-i\lambda)$. Its Minkowski-space integral evaluates to
\begin{align*}
&\tilde F(\epsilon,\rho) =\scalemath{.95}{ (p^2)^{2-\frac{D_0}{2}-\rho+\epsilon}   \int \frac{\d^D k}{(2\pi)^{D}} \frac{i}{(k+p)^2}\frac{i}{ (k^2)^{1-\rho}}(-i\lambda)^2 }\\
&=\scalemath{.82}{\frac{i\lambda^2}{(4\pi)^{\frac{D_0}{2}-\epsilon}} \frac{\Gamma \left( 2-\frac{D_0}{2}-\rho  +\epsilon \right) \Gamma \left( \frac{D_0}{2}-\epsilon -1+\rho \right) \Gamma \left( \frac{D_0}{2}-\epsilon -1 \right)  }{\Gamma \left( D_0-2-2 \epsilon + \rho \right) \Gamma(1-\rho)}}.
\end{align*}
We define the coupling $\alpha$ such that the power in $\alpha$ coincides with the loop number. Hence, $\alpha \propto \lambda^2$. Moreover, we switch to MS-bar instead of MS renormalization conditions in order to absorb powers of $(4\pi)$ and the Euler Mascheroni constant $\gamma_E$ into the coupling, so that
\begin{align}\label{alpha_msbar}
\alpha &:= \frac{\lambda^2}{(4\pi)^{\frac{D_0}{2}}} \left( \frac{4\pi}{e^{\gamma_E }} \right) ^\epsilon .
\end{align}
Notice that if we want to recover the MS-solution from MS-bar, the $\epsilon$ expansion in \cref{mellin_T_series} acquires additional powers of $\gamma_E$ and $\ln(4\pi)$ which are absent in MS-bar. The order $[\epsilon^0]$ is unchanged, consistent with our above result (\cref{lem:MS_anomalous_dimension}) that for a linear DSE, the   anomalous dimension at $\epsilon=0$ is the same in all schemes. The transformation MS $\leftrightarrow$ MS-bar has been spelled out in more detail in  \cite{balduf_dyson_2023}.

The overall factor $i$ in the Mellin transform gets absorbed by the definition of the 1PI self energy $i \Sigma$. Since the full propagator is a geometric series in 1PI propagators, the DSE in \cref{DSE_combinatorial_s} for a propagator 1PI Green's function $G_\ren $ should have a negative sign,  $G_\ren = 1 - \alpha (1-\renop)B_+[G_\ren^{1+s}]$. We will continue using our original definition \cref{DSE_combinatorial_s},  which means that formally  the physical value of $\alpha$ is negative for propagator DSEs
We leave out the prefactor $i\alpha$ from the Mellin transform by setting $i\alpha \tilde F(\epsilon,\rho)=  F(\epsilon,\rho)$, so that now
\begin{align}\label{mellin_multiedge}
&F(\epsilon,\rho) \\
&=  \frac{\Gamma \left( 2-\frac{D_0}{2}-\rho  +\epsilon \right) \Gamma \left( \frac{D_0}{2}-\epsilon -1+\rho \right) \Gamma \left( \frac{D_0}{2}-\epsilon -1 \right)  }{e^{-\gamma_E \epsilon}\cdot \Gamma \left( D_0-2-2 \epsilon + \rho \right) \Gamma(1-\rho)}.\nonumber 
\end{align}

\subsection{Yukawa rainbows}
Massless Yukawa theory  contains fermions $\psi$  and  mesons $\phi$ with an interaction vertex $\lambda \bar \psi \psi \phi$, and is perturbatively renormalizable at $D_0=4$. After projection to the tree level tensor structure, see \cite{delbourgo_dimensional_1996,broadhurst_exact_2001,borinsky_nonperturbative_2020}, the Feynman integral for the fermion propagator coincides with the scalar multiedge. 
Setting $D_0=4$, the Mellin transform \cref{mellin_multiedge} is
\begin{align*}
F(\epsilon; \rho+\epsilon) &= \frac{-  e^{ \gamma_E \epsilon} \pi  \Gamma(1-\epsilon)}{ \sin(\pi \rho)\Gamma(1-\epsilon-\rho) \Gamma(2-\epsilon +\rho) }.
\end{align*}
The functions of \cref{mellin_T_series} are
\begin{align*}
T_0(\rho) &= -\rho\cdot (1+\rho)\\
T_1(\rho) &=  - H_{-\rho} - H_{1+\rho}.
\end{align*}
Here, $H_n=\sum_{k=0}^n \frac 1 k$ is the harmonic number, whose analytic continuation is the digamma function $\psi(z)= \Gamma'(z)/\Gamma(z)$ according to $H_z = \gamma_E + \psi(z+1)$. 

The anomalous dimension at $\epsilon=0$, in all schemes, is   computed from \cref{MOM_algebraic_equation},
\begin{align*}
\alpha &= T_0(\gamma) = -\gamma(1+\gamma), 
\end{align*}
hence we reproduce the  result of \cite{delbourgo_dimensional_1996,kreimer_etude_2008} for the fermion propagator in MOM:
\begin{align}\label{4D_gamma}
G_\ren &= e^{L \gamma(\alpha)},\\
\gamma(\alpha) &= \frac{-1+\sqrt{1-4\alpha}}{2}=-\alpha-\alpha^2-2\alpha^3-5\alpha^4 - \ldots. \nonumber 
\end{align}
The series coefficients are Catalan numbers. We determine the function $g(\alpha)$  from \cref{lem:g_Mellin}, where
\begin{align*}
\partial_\rho T_0 &= -1-2\rho, \qquad 
\frac 12\partial^2_\rho T_0 = -1, \qquad 
  \partial_\alpha \gamma  = \frac{-1}{\sqrt{1-4\alpha}}.
\end{align*}
This leads to an expression with digamma functions $\psi$,
\begin{align*}
g &= \frac{-\alpha }{1-4\alpha} - \alpha \frac{  2 \gamma_E + \psi\left( \frac{3-\sqrt{1-4\alpha}}{2} \right)  + \psi \left( \frac{3+ \sqrt{1-4\alpha}}{2} \right)   }{\sqrt{1-4\alpha}}\\
&= \scalemath{.95}{-2\alpha -7 \alpha^2 +(2 \zeta(3)-26) \alpha^3 + (8 \zeta(3) -99)\alpha^4+\ldots }
\end{align*}
To find the offset function $\bar \gamma_0(\alpha,0)$ according to \cref{thm:DSE_solution_linear}, we need to integrate $g(\alpha)/\alpha$. This integral is easier than it looks because $\psi$ is the derivative of the Euler gamma function. One obtains  
\begin{align}\label{Yukawa_gamma0}
 \bar \gamma_0(\alpha) &= \frac{e^{\gamma_E(1-\sqrt{1-4\alpha})} \Gamma \left( \frac{3-\sqrt{1-4\alpha}}{2} \right)  }{(1-4\alpha)^{\frac 14} \Gamma \left( \frac{3+\sqrt{1-4\alpha}}{2} \right)  } = \frac{e^{-2\gamma_E \gamma} ~\Gamma \left( 1-\gamma \right)  }{\sqrt{1+2\gamma}~ \Gamma \left( 2+\gamma \right)  } \nonumber \\
 &=\scalemath{.9}{ 1 + 2\alpha +\frac{11}{2}\alpha^2 + \left( 17- \frac 2 3 \zeta(3) \right) \alpha^3 + \ldots}. 
\end{align}
This confirms  the formula \cite[eq. (4.15)]{balduf_dyson_2023}, which had been discovered experimentally from matching the first 20 terms of the series expansion. In particular, the general result of \cref{lem:g_Mellin} explains  the empirical observation  that $\bar \gamma_0$ and $\delta$ contain the anomalous dimension $\gamma(\alpha)$ as \enquote{building blocks}. From \cref{thm:DSE_solution_linear}, we find the closed formula
\begin{align}
&\delta(\alpha) =-2 \gamma_E +   \frac{2\ln \left( (1-4\alpha)^{\frac 14} \frac{\Gamma \left( \frac{3+\sqrt{1-4\alpha}}{2} \right)  }{\Gamma\left( \frac{3-\sqrt{1-4\alpha}}{2} \right)  } \right) }{ 1-\sqrt{1-4\alpha}} \\
&= \scalemath{.9}{-2 - \frac 3 2 \alpha + \left( \frac 2 3 \zeta(3) - \frac{19}{6} \right) \alpha^2 +\left( \frac 4 3 \zeta(3) - \frac{103}{12} \right) \alpha^3+ \ldots .}\nonumber 
\end{align}
With these functions, the exact solution of the Yukawa rainbow DSE in minimal subtraction is
\begin{align*}
	\bar G_{ \ren}(\alpha, L) &= \bar \gamma_0(\alpha) e^{L \gamma(\alpha)} = \exp \Big(\big(L+\delta(\alpha)\big)\gamma(\alpha)\Big).
\end{align*}

\subsection{$\phi^3$ rainbows}
In $D_0=6$, the one-loop multiedge, and its corresponding ladder solution shown in \cref{fig:phi3_ladders}, appear as the propagator correction in $\phi^3$ theory. 
We can immediately apply the formalism. The Mellin transform is the specialization of \cref{mellin_multiedge} to $D_0=6$, 
\begin{align}\label{multiedge_D6_T}
F(\epsilon; \rho+\epsilon) &= \frac{ e^{\gamma_E \epsilon}\pi \Gamma(2-\epsilon)}{\sin(\pi \rho) \Gamma(1-\epsilon-\rho) \Gamma(4-\epsilon+\rho)}, \nonumber  \\
T_0 &= \rho( \rho+1)( \rho+2)(\rho+3), \nonumber  \\
T_1 &= 1-H_{-\rho}- H_{ \rho+3}.
\end{align}
The anomalous dimension is
\begin{align}
\gamma(\alpha) &= \frac{-3 + \sqrt{5+4\sqrt{1+\alpha}}}{2},
\end{align}
which again reproduces \cite{delbourgo_dimensional_1997}. To compute the MS solution, \cref{lem:g_Mellin} results in 
\begin{align*}
g&= \scalemath{.9}{\alpha \frac{5 +  \sqrt{1+\alpha} \big(6+  (3+2\gamma) (-1+H_{-\gamma}+  H_{3+\gamma}) \big)}{ 4(1+\alpha)(3+2\gamma)^2}}\\
&=\frac 4 9 \alpha - \frac{535}{1296}\alpha^2 + \left( \frac{9077}{23328}-\frac{\zeta(3)}{108} \right) \alpha^3 + \ldots 
\end{align*}
Inserting this into \cref{thm:DSE_solution_linear}, we obtain closed-form expressions that confirm the experimental finding of \cite[eq. (5.1)]{balduf_dyson_2023}:
\begin{align}
\bar \gamma_0 &= \frac{6 \sqrt 3 e^{ -\gamma(2\gamma_E-1)} ~\Gamma \left( 1-\gamma \right)   }{(1+\alpha)^{\frac 14}  \sqrt{2\gamma+3}~ \Gamma \left( 4+\gamma  \right)  }\\
&= 1-\frac 4 9 \alpha + \frac{791}{2592}\alpha^2 + \left( -\frac{5507}{23328}+ \frac{\zeta(3)}{324} \right) \alpha^3+\ldots, \nonumber \\
\delta &= 1- 2\gamma_E   +\frac{ \ln \frac{6 \sqrt 3  ~\Gamma \left( 1-\gamma \right)   }{(1+\alpha)^{\frac 14}  \sqrt{2\gamma+3}~ \Gamma \left( 4+\gamma  \right)  }}{\gamma}\\
&= -\frac 8 3 + \frac{61}{144} \alpha + \left( -\frac{10493}{46656} + \frac{\zeta(3)}{54} \right) \alpha^2 + \ldots .\nonumber 
\end{align}

\FloatBarrier

\subsection{$\phi^3$ null ladders}

In the present article, we restrict ourselves to Green's functions which depend on only one kinematic variable. For the vertex in $\phi^3$ theory, we achieve this by fixing one of the external momenta to zero, and inserting subdiagrams into the corresponding vertex as shown in \cref{fig:phi3_ladders}. 

\begin{figure}[htb]
	\centering
	\begin{tikzpicture}[scale=.7]
		\node at (-1.1,0){$1+ \alpha$};
		\node[vertex] (v1) at (0,0){};
		\node[vertex](v2) at (30:1){};
		\node[vertex](v3) at (-30:1){};
		
		\draw[edge] (v1) --  (v2);
		\draw[edge] (v1) --  (v3);
		\draw[edge] (v3) --  (v2);
		\draw[edge] (v1)-- +(-.3,0);
		\draw[edge ] (v2)-- +(.3,0);
		\draw[edge ] (v3)-- +(.3,0);
		
		\node at (2.,0){$+\alpha^2$};
		
		\node[vertex] (v1) at (3,0){};
		\node[vertex](v2) at ($(v1)+ (30:.8)$){};
		\node[vertex](v3) at ($(v1)+ (-30:.8)$){};
		\node[vertex](v4) at ($(v1)+ (30:1.6)$){};
		\node[vertex](v5) at ($(v1)+ (-30:1.6)$){};
		
		\draw[edge] (v1) --  (v2);
		\draw[edge] (v1) --  (v3);
		\draw[edge] (v3) --  (v2);
		\draw[edge] (v2) --  (v4);
		\draw[edge] (v3) --  (v5);
		\draw[edge] (v4) --  (v5);
		\draw[edge] (v1)-- +(-.3,0);
		\draw[edge ] (v4)-- +(.3,0);
		\draw[edge ] (v5)-- +(.3,0);
		
		\node at (5.5,0){$+\alpha^3$};
		
		\node[vertex] (v1) at (6.5,0){};
		\node[vertex](v2) at ($(v1)+ (30:.8)$){};
		\node[vertex](v3) at ($(v1)+ (-30:.8)$){};
		\node[vertex](v4) at ($(v1)+ (30:1.6)$){};
		\node[vertex](v5) at ($(v1)+ (-30:1.6)$){};
		\node[vertex](v6) at ($(v1)+ (30:2.4)$){};
		\node[vertex](v7) at ($(v1)+ (-30:2.4)$){};
		
		\draw[edge] (v1) --  (v2);
		\draw[edge] (v1) --  (v3);
		\draw[edge] (v3) --  (v2);
		\draw[edge] (v2) --  (v4);
		\draw[edge] (v3) --  (v5);
		\draw[edge] (v4) --  (v5);
		\draw[edge] (v4) --  (v6);
		\draw[edge] (v5) --  (v7);
		\draw[edge] (v6) --  (v7);
		\draw[edge] (v1)-- +(-.3,0);
		\draw[edge ] (v6)-- +(.3,0);
		\draw[edge ] (v7)-- +(.3,0);
		
		\node at (9.5,0){$+\ldots $};
		
	\end{tikzpicture}
	
	\caption{Sum of ladders in $\phi^3$ theory. }
	\label{fig:phi3_ladders}
\end{figure}
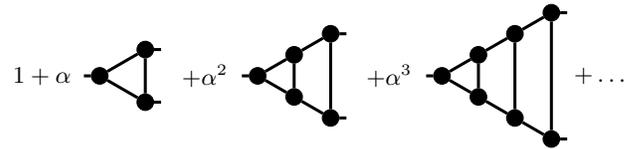

A vertex insertion with zero momentum is equivalent to what would be a mass insertion, that is, a 2-valent vertex that effectively squares the propagator it resides in, see \cref{fig:phi3_ladder_kernel}. The kernel  then amounts to a 1-loop multiedge with propagator powers 1 and $2-\rho$. We obtain its formula by replacing $\rho \mapsto \rho-1$ in \cref{mellin_multiedge} for $D_0=6$. Notice that the multiedge with one squared propagator is not infrared divergent in 6 dimensions. 

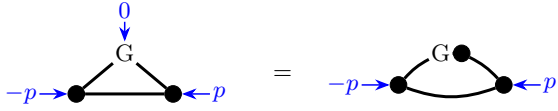
\begin{figure}[htbp]
	\centering
	\begin{tikzpicture}[scale=.7]
		 
		\node[minimum size=3mm,label={[label distance=2mm,blue]above:{$0$}}] (v1) at (0,0){};
		\node at (v1){G};
		\node[vertex,label={[label distance=3mm,blue]left:{$-p$}}](v2) at (220:1.2){};
		\node[vertex,label={[label distance=3mm,blue]right:{$p$}}](v3) at (320:1.2){};
		
		\draw[edge] (v1) --  (v2);
		\draw[edge] (v1) --  (v3);
		\draw[edge] (v3) --  (v2);
		\draw[momentumarrow,<-] (v1)-- +(0,.6);
		\draw[momentumarrow,<- ] (v2)-- +(-.7,0);
		\draw[momentumarrow,<- ] (v3)-- +(.7,0);
		
		\node[] at (3,-.4){$=$};
		
		\node[minimum size=3mm] (v1) at (6,-0){};
		\node at (v1){G};
		\node[vertex] (v4) at ($(v1)+(.4,0)$){};
		\node[vertex,label={[label distance=3mm,blue]left:{$-p$}}](v2) at ($(v1)+(-.8,-.6)$){};
		\node[vertex,label={[label distance=3mm,blue]right:{$p$}}](v3) at ($(v1)+(1.2,-.6)$){};
		
		\draw[edge] (v1) --  (v4);
		\draw[edge, bend angle=10,bend left] (v2) to  (v1);
		\draw[edge,bend angle=10,bend left] (v4) to  (v3);
		\draw[edge, bend angle=20,bend right] (v2) to (v3);
	 
		\draw[momentumarrow,<-] (v2)-- +(-.7,0);
		\draw[momentumarrow,<- ] (v3)-- +(.7,0);

	\end{tikzpicture}
	
	\caption{The triangle diagram with zero momentum transfer is equivalent to a multiedge. Unlike \cref{fig:B_phi3}, the inserted subdiagram does not cancel an adjacent edge, we indicate this by an extra dot.}
	\label{fig:phi3_ladder_kernel}
\end{figure}

The resulting Mellin transform is very similar to \cref{multiedge_D6_T}, namely
\begin{align*}
T_0 &= (\rho-1)\rho(\rho+1)(\rho+2),\\
T_1 &= 1-H_{1-\rho}-H_{2+\rho}.
\end{align*}
The anomalous dimension is 
\begin{align}
\gamma(\alpha) &= \frac{-1+\sqrt{5-4\sqrt{1+\alpha}}}{2}.
\end{align}
This, again, coincides with the MOM result of \cite{delbourgo_dimensional_1997}.
The remaining analysis proceeds as above, one finds
\begin{align}
		\bar \gamma_0 &= \frac{2 e^{-\gamma(2\gamma_E-1)}  \Gamma \left(2-\gamma   \right)  }{(1+\alpha)^{\frac 14}  \sqrt{ 2\gamma+1}~ \Gamma\left( 3+\gamma  \right)  }\\
		&= 1+\alpha+ \frac{39}{32}\alpha^2 + \left( \frac 5 3 - \frac{\zeta(3)}{12} \right) \alpha^3 + \ldots,\nonumber  \\
\delta &=  1-2\gamma_E +\frac{    \ln \left( \frac{2}{(1+\alpha)^{\frac 1 4} \sqrt{2\gamma+1}} \frac{\Gamma \left( 2-\gamma \right) }{\Gamma \left( 3+\gamma \right)  } \right)  }{\gamma}\\
&= \scalemath{.85}{-2-\frac{15}{16}\alpha + \left( \frac{\zeta(3)}{6}-\frac{37}{64} \right) \alpha^2 + \left( -\frac{453}{512}+\frac{\zeta(3)}{12} \right) \alpha^3 +\ldots }.\nonumber 
\end{align}

\subsection{$\phi^4$ rainbows}

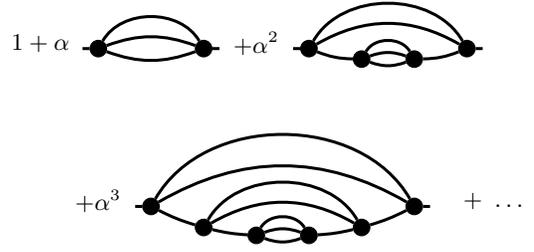
\begin{figure}[htbp]
	\centering
	\begin{tikzpicture}[scale=.7]
		\node at (-1.1, .1){$1 + \alpha$};
		\node[vertex] (v1) at (0,0){};
		\node[vertex](v2) at (2,0){};
		
		\draw[edge,  bend angle=60, bend left] (v1) to (v2);
		\draw[edge,  bend angle=20, bend left] (v1) to (v2);
		\draw[edge,  bend angle=20, bend right] (v1) to (v2);
		\draw[edge] (v1)-- +(-.3,0);
		\draw[edge ] (v2)-- +(.3,0);
		
		\node at (3,.1){$+\alpha^2$};
		
		\node[vertex](v1) at (4,0){};
		\node[vertex](v2) at (7,0){};
		\node[vertex](v3) at (5,-.2){};
		\node[vertex](v4) at (6,-.2){};
		
		\draw[edge,  bend angle=60, bend left] (v1) to (v2);
		\draw[edge,  bend angle=30, bend left] (v1) to (v2);
		\draw[edge,  bend angle=10, bend right] (v1) to (v3);
		\draw[edge,  bend angle=60, bend left] (v3) to (v4);
		\draw[edge,  bend angle=20, bend left] (v3) to (v4);
		\draw[edge,  bend angle=20, bend right] (v3) to (v4);
		\draw[edge,  bend angle=10, bend right] (v4) to (v2);
		\draw[edge] (v1)-- +(-.3,0);
		\draw[edge ] (v2)-- +(.3,0);
		
		\node at (0,-2.9){$+\alpha^3$};

		\node[vertex](v1) at (1,-3){};
		\node[vertex](v2) at (2,-3.4){};
		\node[vertex](v3) at (3,-3.55){};
		\node[vertex](v4) at (4,-3.55){};
		\node[vertex](v5) at (5,-3.4){};
		\node[vertex](v6) at (6,-3){};
		
		\draw[edge,  bend angle=60, bend left] (v1) to (v6);
		\draw[edge,  bend angle=30, bend left] (v1) to (v6);
		\draw[edge,  bend angle=5, bend right] (v1) to (v2);
		\draw[edge,  bend angle=60, bend left] (v2) to (v5);
		\draw[edge,  bend angle=30, bend left] (v2) to (v5);
		\draw[edge,  bend angle=5, bend right] (v2) to (v3);
		\draw[edge,  bend angle=20, bend right] (v3) to (v4);
		\draw[edge,  bend angle=60, bend left] (v3) to (v4);
		\draw[edge,  bend angle=20, bend left] (v3) to (v4);
		\draw[edge,  bend angle=5, bend right] (v4) to (v5);
		\draw[edge,  bend angle=5, bend right] (v5) to (v6);
		
		\draw[edge] (v1)-- +(-.3,0);
		\draw[edge ] (v6)-- +(.3,0);
		
		\node at (7.5,-2.9){$+~\ldots$};
	\end{tikzpicture}
	
	\caption{Rainbows in $\phi^4$ theory. The kernel is a 2-loop diagram, which is primitive since massless tadpoles vanish. }
	\label{fig:phi4_rainbows}
\end{figure}

The kernel diagram for the DSE may have arbitrary loop order, but it must not have subdivergences. For the propagator in $\phi^4$ theory, the leading diagram is a 2-loop multiedge (\enquote{sunrise}), whose 1-loop multiedge subdiagrams are UV-divergent. However, when the subdivergence is replaced by a counterterm, one obtains a 1-loop tadpole, which vanishes in a massless theory. Therefore, the 2-loop multiedge is indeed primitive and we can use our formalism for the version of rainbows shown in \cref{fig:phi4_rainbows}. 
The Mellin transform of the 2-loop multiedge is quite similar to that of the 1-loop multiedge in \cref{mellin_multiedge}, namely
\begin{align*}
	\frac{\Gamma \left( 3-\rho - D_0 + 2\epsilon \right)  \Gamma \left( \frac{D_0}{2}-\epsilon-1 \right)^2   \Gamma \left( \frac{D_0}{2}-\epsilon -1+\rho \right)    }{e^{-2\gamma_E \epsilon}\Gamma \left( \frac 3 2 D_0 - 3 \epsilon - 3 + \rho  \right) \Gamma \left( 1-\rho \right)   }.
\end{align*}
We are interested in $D_0=4$. Since this is a 2-loop diagram,   we now need $F(\epsilon, \rho-2\epsilon)$ to compute $T_1$ of \cref{mellin_T_series}, the result is
\begin{align*}
	T_0 & = \rho (\rho+1)^2 (\rho+2),   \\
	T_1 &= - \frac{1}{\rho+1} - \frac{1}{\rho+2} - 2H_{-\rho}-2H_\rho. 
\end{align*}
Notice that $T_0$ contains a squared factor, which was not the case in any other model considered. This difference leads to interesting consequences for a resurgence analysis of a non-linear  DSEs in \cite{borinsky_treetubings_2024}. In our case of a linear DSE, the anomalous dimension is
\begin{align}
	\gamma(\alpha) &= \sqrt{\frac{1+\sqrt{1+4\alpha}}{2}}-1.
\end{align}
Application of our formalism delivers
\begin{align*}
\bar \gamma_0 &= \scalemath{.87}{\frac{ 2^{\frac 74}\left( ( 1+4\alpha) ( 2(1+\gamma)^2 + 2 \alpha(1+(1+\gamma)^2))\right) ^{-\frac 14}   ~ \Gamma \left( 1-\gamma  \right) ^2}{  e^{4 \gamma \gamma_E} \sqrt{6+4\gamma+2(\gamma+1)^2}~\Gamma \left( 1+\gamma \right) ^2}}\\
&= \scalemath{.73}{1-2\alpha + \frac{175}{32}\alpha^2 + \left( -\frac{539}{32}+\frac{\zeta(3)}{6} \right) \alpha^3 +\left( \frac{113127}{2048}-\frac{23 \zeta(3)}{24} \right) \alpha^4+ \ldots}
\end{align*}
Unlike previous examples, this $\bar \gamma_0$ contains squares of Euler gamma functions. Taking the logarithm, as in all other cases, gives a closed formula for $\delta = \frac{\ln \bar \gamma_0}{\gamma}$ which can readily be obtained with a computer algebra system, its series starts with
\begin{align*}
\delta &= \scalemath{.82}{-4 + \frac{31}{16}\alpha + \left( -\frac{811}{192}+ \frac{\zeta(3)}{3} \right) \alpha^2 + \left( \frac{18071}{1536}- \frac{5\zeta(3)}{6} \right) \alpha^3+\ldots. }
\end{align*} 

\bigskip 

Our list of examples is not exhaustive, further single-kernel DSEs that can be solved with this formalism have appeared e.g. in  \cite{broadhurst_renormalization_1999}. We merely remark that certain models that immediately come to mind, such as null-boxes in $\phi^4$ theory, require extra care because of potential infrared divergences. We leave the treatment of IR-divergences for future work.

\begin{acknowledgments}
	
	I thank Karen Yeats for discussions and comments on the draft. 
	
	Parts of this work were funded through the Royal Society grant {URF{\textbackslash}R1{\textbackslash}201473}. Other parts were carried out while the author was affiliated with Humboldt-Universität zu Berlin. 
\end{acknowledgments}

\bibliography{ladders_rainbows.bib}

\end{document}